\date{\today}
\title{ On The Brownian Loop Measure}
\author{Yong Han \footnote{Institute of Mathematics, Academy of Mathematics and Systems Sciences, Chinese Academy of
Sciences and University of Chinese Academy of
Sciences, Beijing 100190, China. MAPMO, Universit\'e d'Orl\'eans Orl\'eans Cedex 2, France. Email: hanyong@amss.ac.cn}\and Yuefei Wang\footnote{Institute of Mathematics, Academy of Mathematics and Systems Sciences, Chinese Academy of
Sciences and University of Chinese Academy of
Sciences, Beijing 100190, China  Email: wangyf@math.ac.cn}
 \and Michel Zinsmeister\footnote{MAPMO, Universit\'e d'Orl\'eans Orl\'eans Cedex 2, France
     Email:zins@univ-orleans.fr}}
\documentclass[12pt]{article}
\usepackage{amsmath}
\usepackage{amssymb}
\usepackage{amsthm}
\usepackage{amsfonts}
\usepackage{graphicx}
\usepackage{verbatim}

\usepackage{url}

\newif\ifdraft
\drafttrue \long\def\note#1/{\ifdraft{\marginpar{{$\Longleftarrow$}} \bf [#1] }\fi}
\long\def\comment#1{} \long\def\old#1{}

\numberwithin{equation}{section}
\numberwithin{figure}{section}

\newtheorem{theorem}{Theorem}

\newtheorem{lemma}[theorem]{Lemma}

\theoremstyle{remark}
\theoremstyle{remark}\newtheorem{remark}[theorem]{Remark}

\let\qqed=\qed
\def\QED{\qqed\medskip}
\let\qed=\QED

\newcommand{\R}{\mathbb{R}}
\newcommand{\C}{\mathbb{C}}

\def\H{\mathbb{H}}
\def\D{\mathbb{D}}

\def\Im{{\rm Im}\,}
\def\Re{{\rm Re}\,}

\def\SLEkk#1/{$\mathrm{SLE}(#1)$}
\def\SLEr#1/{$\mathrm{SLE(\kappa;#1)}$}
\def\SLEkr#1;#2/{$\mathrm{SLE(#1;#2)}$}
\def\SLEk/{\SLEkk{\kappa}/}
\def\SLEtwo/{\SLEkk2/}
\def\SLE/{$\mathrm{SLE}$}
\def\standardg/{standard Gaussian}

\def\Ito/{It\^o}

\def \prob {{\bf P}}

\def\noopsort#1{}

\def\HY{\,_2F_1}
\def\hy{\,_3F_2}
\begin{document}

\maketitle
\begin{abstract}
In 2003 Lawler and Werner introduced the Brownian loop measure and studied some of its properties. Cardy and Gamsa has predicted  a formula
for the total mass of the Brownian loop measure on the set of simple loops in the upper half plane and disconnect two
given points from the boundary.
In this paper we give a rigorous proof of the formula using a result by Beliaev and Viklund and heavy computations.

\end{abstract}

\section{Introduction}

 The conformally  invariant scaling limits of  a series of planar lattice models  can be described by
 the one-parameter family of random fractal curves  \SLEkk\kappa/, which was introduced by Schramm. These models include
 site percolation on the triangular graph, loop erased random walk, Ising model, harmonic random walk, discrete Gaussian free field, FK-Ising model and uniform spanning tree. Using \SLE/
as a tool, many problems related to the properties of these models have been solved, such as  the arm exponents for these
models. There are also some variants
of \SLE/ (conformal loop ensemble, Brownian loop measure, Brownian bubble measure) that describe the scaling limit of the random  loops in these models. Therefore it is natural  to use \SLE/ to get properties of these loop measures. One of these applications is to use \SLEkk\frac{8}{3}/ to study the properties of the Brownian bubble measure and Brownian loop measure.
In fact, by  rescaling and letting the two  end points tend to one common point, one can get the Brownian bubble measure(up to  multiplication by a constant).

Recently Beliaev and Viklund \cite{beliaev2013some} found a formula for the probability that
two given points  lies to the left of the \SLEkk\frac{8}{3}/ curve and used it to study some connectivity functions for
\SLEkk\frac{8}{3}/ bubbles and to reconstruct the chordal restriction measure introduced by Lawler, Werner and Schramm \cite{lawler2003conformal}.
In this paper, we will follow their work and use the \SLEkk\frac{8}{3}/ bubble measure to derive the formula for the total mass of  the Brownian loop that disconnects two given points from the boundary. This formula was predicted by Cardy and Gamsa \cite{gamsa2006correlation}, and the formula we get here just differs by an additive constant from theirs.

In the following section, we  give a brief introduction to the topics that will be used in this paper, which include
the definition of \SLE/ processes, Brownian bubble measure, Brownian loop measure, \SLEkk\kappa/ bubble measure and the relation between these measures. In the third section, we give the proof of the main theorem.

\begin{theorem}
Denote by $\mu_\H^\text{loop}$  the Brownian loop measure on the upper half plane and by $\gamma$  a sample of the Brownian loop. Given two points
$z=x+iy, w=u+iv\in\H$, let $E(z,w)$ denote the event that $\gamma$ disconnects both $z$ and $w$ from the boundary of $\H$. Then we have
\begin{align}\label{cardy}
\begin{split}
\mu_\H^\text{loop}[E(z,w)]=&-\frac{\pi}{5\sqrt{3}}-\frac{1}{10}\eta\,_3F_2(1,\frac{4}{3},1;\frac{5}{3},2;\eta)-\frac{1}{10}\log(\eta(\eta-1))\\[4mm]
&+\frac{\Gamma(\frac{2}{3})^2}{5\Gamma(\frac{4}{3})}(\eta(\eta-1))^{\frac{1}{3}}\,_2F_1(1,\frac{2}{3};\frac{4}{3},\eta).
\end{split}
\end{align}
where
\begin{equation}\label{eta}
\eta=\eta(z,w)=-\frac{(x-u)^2+(y-v)^2}{4yv},
\end{equation}
 and $\,_3F_2,\,_2F_1$ are the hypergeometric functions.
\end{theorem}
 The formula \eqref{cardy} was first given by Cardy and Gamsa using  conformal field theory which assumes that $O(n)$ model has the scaling limit.
In fact \eqref{cardy} has a nicer form:
\begin{equation}\label{han}
\mu^\text{loop}[E(z,w)]=-\frac{1}{10}[\log\sigma+(1-\sigma)\,_3F_2(1,\frac{4}{3},1;\frac{5}{3},2;1-\sigma)],
\end{equation}
where
\begin{equation}\label{sigma}
\sigma=\sigma(z,w)=\frac{|z-w|^2}{|z-\bar{w}|^2}=\frac{(x-u)^2+(y-v)^2}{(x-u)^2+(y+v)^2},
\end{equation}
and $\,_3F_2$ is the hypergeometric function.
\begin{remark}
By the conformal invariance of Brownian loop measure (see \cite{lawler2004brownian}), for any simply connected domain $D\subset\C$ with $z,w\in D$, we can get the total mass of the Brownian loop in $D$ that disconnect both $z$ and $w$ from $\partial D$ by the conformal map from $D$ to $\H$. In particular, if $D=\D$, we choose the conformal map $\phi(z)=i\frac{1+z}{1-z}$ from $\D$ onto $\H$. Then the total mass of the Brownian loop measure in $\D$ that disconnects $z,w\in\D$ from $\partial\D$ is
$$
-\frac{1}{10}[\log\tilde{\sigma}+(1-\tilde{\sigma})\,_3F_2(1,\frac{4}{3},1;\frac{5}{3},2;1-\tilde{\sigma})],
$$
where $\tilde{\sigma}=\tilde{\sigma}(z,w)=\frac{|z-w|^2}{|1-z\bar{w}|^2}$.
\end{remark}
\begin{remark}
In fact, in Cardy and Gamsa's  original paper, their formula is just \eqref{cardy} without the constant term and this is a tiny mistake in their paper.
\end{remark}
\section{Background}
In this section  very brief introductions of chordal \SLE/, Brownian loop measure and {\SLE/ bubble measure} are given.
\subsection{Chordal \SLE/ process}
The chordal \SLE/ process from $0$ to $\infty$ in $
\H$ is the  random family of conformal maps $(g_t:t\geq 0)$ that satisfies
\begin{equation}\label{chordal}
\partial_tg_t(z)=\frac{2}{g_t(z)-\sqrt{\kappa}B_t},\,\,\,0\leq t<\tau(z),\,\,\,\,g_0(z)=z.
\end{equation}
where $B$ is a standard Brownian motion and $\tau(z)$ is the blow-up time for the differential equation \eqref{chordal}.
The \SLEkk\kappa/ process is generated by a continuous curve $\gamma$, (see \cite{rohde2011basic} and \cite{lawler2011conformal}). That is, the following limits exist:
$$
\gamma(t)=\lim\limits_{y\downarrow 0}g_t^{-1}(iy+\sqrt{\kappa}B_t).
$$
The curve $\gamma$ is a simple curve if and only if $\kappa\in(0,4]$, (see\cite{rohde2011basic}). The \SLEkk\kappa/ curve satisfies the scaling invariance property (see \cite{rohde2011basic}), i.e for any $r>0$, $r\gamma$ has the same distribution as $\gamma$ (with a time rescaling). So given any triple set $(D,a,b)$, where $D$ is a simply connected domain with two given boundary points $a$ and $b$, we  can define the \SLE/ process on $D$ from $a$ to $b$ by the conformal map from $\H$ to $D$ that sends $0$ to $a$ and $\infty$ to $b$.

The following
lemma will be useful in defining the \SLE/-bubble measure.
\begin{lemma}[see \cite{schramm2001percolation}]\label{lemma4}
Given $z=x+iy\in\H$, suppose $\kappa\in(0,4]$ and  $\gamma$ is the \SLEkk\kappa/ curve from $0$ to $\infty$ in $\H$. Then the
probability that  $\gamma$ passes the left of $z$ is
\begin{equation}\label{loop111}
p(z)=C\int_{-\infty}^{\frac{x}{y}}(1+t^2)^{-\frac{4}{\kappa}}dt,
\end{equation}
where $C=C(\kappa)$ is the constant that make the total integral above equal to $1$.
\end{lemma}

\subsection{Brownian loop measure and  bubble measure}
In this section, we will introduce several measures on the space of continuous curves in the plane. To keep the present article short,  we will not provide the detailed  discussions but instead refer the reader to the fifth chapter of Lawler's
book \cite{lawler2008conformally} and \cite{lawler2004brownian}.

Let $\mu(z;t)$ be the law of a complex Brownian motion $(B_s:0\leq s\leq t)$ starting from $z$. It can be written as
$$
\mu(z;t)=\int_{\C} \mu(z,w;t)dw
$$
where the above integral may be regarded as the integral of measure-valued functions. Here $\mu(z;t)$ and $\mu(z,w;t)$ are regarded as measures on the space of curves in the plane.
Using the density function of the complex Brownian motion, we can see that  the total mass of $\mu(z,w;t)$ is $\frac{1}{2\pi t}\exp\{-\frac{1}{2t}|z-w|^2\}$.

Let $\mu(z,w)$ be the measure defined by
$\mu(z,w)=\int_0^\infty \mu(z,w;t)$. This is a $\sigma-$finite infinite measure. If $D\subset\C$ is a simply connected
domain with nice boundary and $z,w\in D$, we can define $\mu_D(z,w)$ be the restriction of $\mu(z,w)$ on the space of curves that lie inside $D$. If $z\neq w$, the total measure of $\mu_D(z,w)$ is  $\pi G_D(z,w)$, where $G_D(z,w)$ is the Green function on $D$.

If $D$ is a simply connected domain with nice boundary, let $B$ be a complex Brownian motion starting from $z\in D$ and $\tau_D$ the exit time. Denote $\mu_D(z,\partial D)$ the law of $(B_t:0\leq t\leq \tau_D)$, we can write
$$
\mu_D(z,\partial D)=\int_{\partial D}\mu_D(z,w)dw.
$$
Here  $\mu_D(z,w)$  can be regarded as a measure on the space of curves in $D$ from $z$ to $w\in\partial D$, and the total mass
of $\mu_D(z,w)$ is the the Poisson kernel $H_D(z,w)$.
For $z\in D, w\in\partial D$,  $\mu_D(z,w)$ can also be equivalently defined by  the limits
$$
\mu_D(z,w)=\lim\limits_{\epsilon\rightarrow 0}\frac{1}{2\epsilon}\mu_D(z,w+\epsilon\textbf{n}_w),
$$
where $\textbf{n}_w$ is the inner normal  at $w$.

 And similarly for $z,w\in\partial D$, we can also define
$$
\mu_D(z,w)=\lim\limits_{\epsilon\rightarrow 0}\frac{1}{2\epsilon^2}\mu_D(z+\epsilon\textbf{n}_z,w+\epsilon\textbf{n}_w)
$$
It can be showed that the above limits exist in the sense of Prohorov convergence(see Chapter 5 of \cite{lawler2008conformally}).

Given $z\in \partial D$, the \textbf{Brownian bubble measure }$\mu_D^{\text{bub}}(z)$ is defined as the limit
$$
\mu_D^{\text{bub}}(z):=\lim\limits_{w\in\partial D, w\rightarrow z}\pi \mu_D(z,w).
$$
The \textbf{Brownian loop measure} is defined as following:
$$
\mu^{loop}_\C:=\int_\C \frac{1}{t_\gamma}\mu(z,z)dz=\int_\C\int_0^\infty \frac{1}{t_\gamma}\mu(z,z;t)dtdz.
$$
Since $\mu(z,z)$ is a measure defined on loops with $z$ as a marked point(called a root), the  Brownian loop measure should
be understood as the above integral of measures by forgetting the root. For any  domain $D$, let $\mu_D^{\text{loop}}$ be the restriction of the Brownian loop measure on the space of loops inside $D$.

For any $a\in\R$, define $\H_a:=\{x+iy\in\C:y>a\}$, according to the lowest point of the Brownian loop, the Brownian loop   can be decomposed into the following integral  of Brownian bubbles(see \cite{lawler2004brownian}):
\begin{equation}\label{loop10}
\mu_\C^{\text{loop}}=\frac{1}{\pi}\int_\C\mu_{\H_y}^{\text{bub}}(x+iy)dxdy.
\end{equation}
\eqref{loop10} will be very important in our computation.

\subsection{\SLE/ bubble measure}
In this section we will define the \SLE/ bubble measure and give the relation between \SLEkk\frac{8}{3}/ bubble and Brownian bubble measure.

Suppose $\kappa\in(0,4], \epsilon>0$ and $\gamma^\epsilon$ is the \SLEkk\kappa/ curve from $0$ to $\epsilon$ in the upper half plane.  Let $\mu^{\epsilon}$ denote the law of $\gamma^\epsilon$.
\begin{lemma}
The  following limit exists:
\begin{equation}
\mu_{\text{SLE}(\kappa)}^{\text{bub}}(0)=\lim\limits_{\epsilon\rightarrow 0}\epsilon^{1-\frac{8}{\kappa}}\mu^{\epsilon}.
\end{equation}
We call $\mu_{\text{SLE}(\kappa)}^{\text{bub}}(0)$ the \textbf{\SLEkk\kappa/-bubble measure}.
\end{lemma}
\begin{proof}
We only need to show that the limit restricted to some generated algebras that consist of finite mass exists.
Here we choose the measurable sets $\{\gamma:\gamma\text{ disconnects } z\text{ from }\infty\}$ for fixed $z\in \H$.
By the definition of the \SLEkk\kappa/ from $0$ to $\epsilon$, we choose the auto-conformal map $F_\epsilon(z)=\frac{\epsilon z}{z+1}$ that sends $\infty$ to $\epsilon$ and fixes $0$.  We have
$$
\prob[\gamma\text{ disconnects } z\text{ from }\infty]=p(F_\epsilon^{-1}(z)),
$$
where $p(z)$ is defined in \eqref{loop111}. Therefore
\begin{equation}\label{loop12}
\lim\limits_{\epsilon\rightarrow 0}\epsilon^{1-\frac{8}{\kappa}}p(F_\epsilon^{-1}(z))=\frac{\Gamma(\frac{4}{\kappa})}{\sqrt{\pi}\Gamma(\frac{8-\kappa}{2\kappa})(\frac{8}{\kappa}-1)}(\frac{x^2+y^2}{y})^{1-\frac{8}{\kappa}}
\end{equation}
So for fix $z\in\H$, if we denote $\mu^{\epsilon}(z)$ the restriction of $\mu^\epsilon$ restricted to the curves that disconnect $z$ from $\infty$, then by the above equation, we know that the limit
$$
\mu_{\text{SLE}(\kappa)}^{\text{bub}}(0,z):=\lim\limits_{\epsilon\rightarrow 0}\epsilon^{1-\frac{8}{\kappa}}\mu^{\epsilon}(z)
$$
exists and therefore we can define $\mu_{\text{SLE}(\kappa)}^{\text{bub}}(0)$ as the limit of $\mu_{\text{SLE}(\kappa)}^{\text{bub}}(0,z)$
as $z$ tends to zero.
\end{proof}
 If $\kappa=\frac{8}{3}$, from \eqref{loop12}, we get that the total mass of the \SLEkk\frac{8}{3}/-bubble that disconnects
 a given point $z=x+iy\in\H$ is $\frac{1}{4}(\frac{y}{x^2+y^2})^2=\frac{1}{4}(\Im\frac{1}{z})^2$ which corresponds to the
 part (a) of proposition $3.1$ in \cite{beliaev2013some}. In fact, \cite{beliaev2013some} also gives the measure of the \SLEkk\frac{8}{3}/-bubble that disconnects
 two points $z,w\in\H$  from $\infty$ which we will state as the following lemma.
\begin{lemma}[see \cite{beliaev2013some}]\label{lemma6}
Let $E(z,w)$  be the event that two  points $z,w\in\H$ are disconnected from $\infty$ by a \SLEkk\frac{8}{3}/ curve from $0$ to $\epsilon$, then
\begin{equation}
\mu^\epsilon[E(z,w)]=\frac{1}{4}\Im(\frac{1}{z})\Im(\frac{1}{w})G(\sigma(z,w))\epsilon^2+O(\epsilon^3).
\end{equation}
where
$\sigma$ is defined as \eqref{sigma} and
\begin{equation}\label{loop11}
G(t)=1-t\,_2F_1(1,\frac{4}{3};\frac{5}{3};1-t).
\end{equation}
Here $\,_2F_1$ is the hypergeometric function.
\end{lemma}

Notice that when $\kappa=\frac{8}{3}$, it holds that $1-\frac{8}{\kappa}=-2$, so we have
\begin{equation}\label{characterization}
\mu_{\text{SLE}}^{\text{bub}}(0)[E(z,w)]=\frac{1}{4}\Im(\frac{1}{z})\Im(\frac{1}{w})=\frac{1}{4}\frac{yv}{(x^2+y^2)(u^2+v^2)}G(\sigma(z,w)).
\end{equation}

\SLEkk\frac{8}{3}/-bubble measure is  closely related to Brownian bubble, in fact they only differ by a constant. Given a loop $\gamma$ such that $\gamma(0)=\gamma(t_\gamma)=0$ and $\gamma(0,t_\gamma)\subset\H$, call
the complement of the unbounded connected component of $\H\backslash\gamma$ the hull enclosed by $\gamma$. In the following we will let $\gamma$ denote both the loop and hull enclosed by it.
\begin{lemma}[see \cite{lawler2003conformal}]\label{lemma5}
As measures on the hulls enclosed by loops, the following holds:
$$
\mu_\H^{\text{bub}}(0)=\frac{8}{5}\mu_{\text{SLE}(\frac{8}{3})}^{\text{bub}}(0).
$$
\end{lemma}
\begin{proof}
By the construction of the Brownian bubble measure at $0$ (see the Chapter 5 of Lawler's book \cite{lawler2008conformally}), it is the
unique measure on loops(or hulls enclosed by loops) in $\H$ rooted at $0$ such that  the  total mass  of the set of loops intersecting $\{z: |z|=r\}$ is
$\frac{1}{r^2}$ for any $r>0$.
Since
$$
\big\{\{\gamma:\gamma\cap \{|z|=r\}=\emptyset\}:r>0\big\}
$$
is an algebra that generate the $\sigma-$algebra of the space of loops,
we only need to show that the total mass of the \SLEkk\frac{8}{3}/-bubble
sample intersecting $|z|=r$ is $\frac{5}{8r^2}$.
Define $F_\epsilon(z)=\frac{z}{\epsilon-z}$; the image of the circle $|z|=r$ under $F_\epsilon$ is a circle
with center $c_0=-\frac{r^2}{r^2-\epsilon^2}$ and radius $\rho=\frac{\epsilon r}{r^2-\epsilon^2}$. Define the
conformal map $\phi_\epsilon(z)=z-c_0+\frac{\rho^2}{z-c_0}$ which maps $\H\backslash B(c_0,\rho)$ onto $\H$ with the derivative at $\infty$ equaling to $1$. By the conformal restriction property of \SLEkk\frac{8}{3}/, we have
$$
\mu^\epsilon[\gamma\cap |z|=r=\emptyset]=\mu^\infty[\gamma\cap B(c_0,\rho)=\emptyset]=\phi'_\epsilon(0)^{\frac{5}{8}}.
$$
Therefore we can check that
$$
\mu_{\text{SLE}(\kappa)}^{\text{bub}}(0)[\gamma\cap \{|z|=r\}\neq\emptyset]=\lim\limits_{\epsilon\rightarrow 0}\frac{1}{\epsilon^2}(1-\phi_\epsilon'(0))^{\frac{5}{8}}=\frac{5}{8r^2}.
$$
\end{proof}



\section{Proof of the main theorem}
In this section we will give a detailed proof of our main theorem.
\subsection{Proof of the nicer form}
Suppose we are given two points
$z_0=x_0+iy_0$ and $w_0=u_0+iv_0\in\H$. By the symmetry property of the the Brownian loop measure, we  may assume without loss of generality that $y_0\leq v_0$, $u_0\geq x_0$.  By \eqref{loop10},
\begin{align*}
\mu^{\text{loop}}_\H[E(z_0,w_0)]&=\frac{1}{\pi}\int_{\H}\mu_{\H_y}^{\text{bub}}(x+iy)[E(z_0,w_0)]dxdy\\[3mm]
&=\frac{1}{\pi}\int_0^{y_0}
\int_{\R}\mu_{\H_y}^{\text{bub}}(x+iy)[E(z_0,w_0)]dxdy.
\end{align*}
Here $E(z_0,w_0)$ denotes the event that the Brownian loop sample in $\H$ disconnects both $z_0$ and $w_0$ from the boundary of $\H$.

By the translation invariance of the Brownian bubble measure, we have
$$
\mu_{\H_y}^{\text{bub}}(x+iy)[E(z_0,w_0)]=\mu_{\H}^{\text{bub}}(0)[E(z_0-z,w_0-z)].
$$
By Lemma \ref{lemma5} we have $\mu_{\H}^{\text{bub}}(0)=\frac{8}{5}\mu_{\text{SLE}(\kappa)}^{\text{bub}}(0)$. Therefore by \eqref{characterization},
\begin{multline}\label{loop1}
\mu^{\text{loop}}_\H[E(z_0,w_0)]=\frac{8}{5\pi}\int_0^{y_0}\int_\R\mu_{\text{SLE}(\kappa)}^{\text{bub}}(0)[E(z_0-x-iy,w_0-x-iy)]dxdy\\
=\frac{8}{5\pi}\int_0^{y_0}\int_\R\frac{1}{4}\Im\Big(\frac{1}{z_0-x-iy}\Big)\Im\Big(\frac{1}{w_0-x-iy}\Big)\\
G(\sigma(z_0-x-iy,w_0-x-iy))dxdy.
\end{multline}

So  in order to prove the  theorem, we only need to compute above integral. Define two functions as follows:
\begin{equation}
f(x,y):=\frac{(y_0-y)(v_0-y)}{[(x_0-x)^2+(y_0-y)^2]\times[(u_0-x)^2+(v_0-y)^2]}.
\end{equation}
\begin{multline}
g(y):=\frac{(x_0-u_0)^2+(y_0-v_0)^2}{(x_0-u_0)^2+(y_0+v_0-2y)^2}\\
\,_2F_1(1,\frac{4}{3};\frac{5}{3};\frac{4(y_0-y)(v_0-y)}{(x_0-u_0)^2+(y_0+v_0-2y)^2}).
\end{multline}

\begin{lemma}
With the notations as above, for fixed $y>0$,
\begin{equation}
\int_\R f(x,y)dx=\frac{2(y_0-y)+v_0-y_0}{(x_0-u_0)^2+\big(2(y_0-y)+v_0-y_0\big)^2}\pi.
\end{equation}
\end{lemma}
\begin{proof}
For fixed $y>0$,  denote
$$
a=y_0-y,\,\,\,\,b=v_0-y,\,\,\,\,c=u_0-x_0,\,\,\,\,d=v_0-y_0.
$$
Then we have
$$
f(x,y)=\frac{ab}{[(x_0-x)^2+a^2][(u_0-x)^2+b^2]}.
$$
By standard calculus,
\begin{align*}
&\int_\R f(x,y)dx=\int_\R \frac{ab}{[(x_0-x)^2+a^2][(u_0-x)^2+b^2]}dx\\
=&ab\int_\R \frac{1}{[x^2+a^2][(x+c)^2+b^2]}dx\\
=&\frac{ab\pi}{ab\big(a^4-2a^2(b^2-c^2)+(b^2+c^2)^2\big)}\Big(b(b^2+c^2-a^2)\arctan[\frac{x}{a}]\\
&+a\big[(a^2+c^2-b^2)\arctan[\frac{c+x}{b}]+bc\log\frac{b^2+(c+x)^2}{a^2+x^2}\big]\Big)|_{-\infty}^{\infty}\\
=&\pi\frac{b(b^2+c^2-a^2)+a(a^2+c^2-b^2)}{a^4-2a^2(b^2-c^2)+(b^2+c^2)^2}.\\
\end{align*}
Replace $b$ by $a+d$, we  get
$$
\int_\R f(x,y)dx=\frac{(2a+d)\pi}{c^2+(2a+d)^2},
$$
which is what we want.
\end{proof}

By \eqref{loop1}, we have
\begin{align}\label{loop44}
&\mu^{\text{loop}}_\H[E(z_0,w_0)]=\frac{8}{5\pi}\int_0^{y_0}\int_\R\frac{1}{4}f(x,y)(1-g(y))dxdy\nonumber\\[3mm]
=&\frac{8}{5\pi}\int_{0}^{y_0} \frac{\pi}{4}\frac{2a+d}{c^2+(2a+d)^2}[1-g(y)]dy
=\frac{2}{5}(A-B),
\end{align}
where
\begin{equation}\label{loop2}
A=A(z_0,w_0)=\int_0^{y_0}\frac{2(y_0-y)+v_0-y_0}{(x_0-u_0)^2+\big(2(y_0-y)+v_0-y_0\big)^2}dy,
\end{equation}
and
\begin{equation}\label{loop3}
B=B(z_0,w_0)=\int_0^{y_0}\frac{2(y_0-y)+v_0-y_0}{(x_0-u_0)^2+\big(2(y_0-y)+v_0-y_0\big)^2}g(y)dy.
\end{equation}
\begin{lemma}\label{lemma1}
$$
A=\frac{1}{4}\log\frac{1}{\sigma},
$$
where $\sigma$ is defined as \eqref{sigma}.
\end{lemma}
\begin{proof}
By \eqref{loop2} we have
\begin{align*}
A=&\int_0^{y_0}\frac{2(y_0-y)+d}{c^2+(2(y_0-y)+d)^2}dy
=\int_0^{y_0}\frac{2y+d}{c^2+(2y+d)^2}dy\\[3mm]
=&\frac{1}{2}\int_{d/c}^{\frac{2y_0+d}{c}}\frac{y}{1+y^2}dy=\frac{1}{4}\log\frac{c^2+(2y_0+d)^2}{c^2+d^2}.
\end{align*}
In the second equation we have  used the  change of variable $y\rightarrow y_0-y$ and in the last equation the
change of variable $y\rightarrow \frac{2y+d}{c}$. Notice that
\begin{align*}
\frac{c^2+(2y_0+d)^2}{c^2+d^2}=\frac{(u_0-x_0)^2+(y_0+v_0)^2}{(u_0-x_0)^2+(v_0-y_0)^2}=\frac{1}{\sigma}.
\end{align*}
\end{proof}
\begin{lemma}\label{lemma2}
$$
B=\frac{1}{4}(1-\sigma)\,_3F_2(1,\frac{4}{3},1;\frac{5}{3},2;1-\sigma),
$$
where $\sigma$ is defined as \eqref{sigma}.
\end{lemma}
\begin{proof}
By \eqref{loop3}  and the definition of $g(y)$, we have
\begin{align*}
B=&\int_0^{y_0}\frac{2(y_0-y)+d}{c^2+(2(y_0-y)+d)^2}\cdot\frac{c^2+d^2}{c^2+(2(y_0-y)+d)^2}\cdot\\[5mm]
&\,_2F_1(1,\frac{4}{3};\frac{5}{3};\frac{4(y_0-y)((y_0-y)+d)}{c^2+(2(y_0-y)+d)^2})dy\\[5mm]
=&\int_0^{y_0}\frac{2y+d}{c^2+(2y+d)^2}\cdot\frac{c^2+d^2}{c^2+(2y+d)^2}\cdot\,_2F_1(1,\frac{4}{3};\frac{5}{3};\frac{4y(y+d)}{c^2+(2y+d)^2})dy\\[5mm]
=&\int_{\frac{d}{c}}^{\frac{2y_0+d}{c}}\frac{cy}{c^2+c^2y^2}\cdot\frac{c^2+d^2}{c^2+c^2y^2}\cdot\,_2F_1(1,\frac{4}{3};\frac{5}{3};\frac{c^2y^2-d^2}{c^2+c^2y^2})\cdot\frac{c}{2}dy\\[5mm]
=&\frac{1}{2}\frac{c^2+d^2}{c^2}\int_{\frac{d}{c}}^{\frac{2y_0+d}{c}}\frac{y}{(1+y^2)^2}\cdot\,_2F_1(1,\frac{4}{3};\frac{5}{3};\frac{c^2y^2-d^2}{c^2+c^2y^2})dy\\[5mm]
=&\frac{1}{4}\int_0^{\frac{4y_0(y_0+d)}{c^2+(2y_0+d)^2}}\,\,_2F_1(1,\frac{4}{3};\frac{5}{3};y)dy\\[5mm]
=&\frac{1}{4}\frac{4y_0(y_0+d)}{c^2+(2y_0+d)^2}\cdot \,_3F_2(1,\frac{4}{3},1;\frac{5}{3},2;\frac{4y_0(y_0+d)}{c^2+(2y_0+d)^2})\\[5mm]
=&\frac{1}{4}(1-\sigma)\,_3F_2(1,\frac{4}{3},1;\frac{5}{3},2;1-\sigma).
\end{align*}
Here the second equation used the change of variable $y\rightarrow y_0-y$, the third equation used the change of variable $y\rightarrow \frac{2y+d}{c}$, the fifth equation used the change of variable
 $\frac{c^2y^2-d^2}{c^2+c^2y^2}\rightarrow y$ and the sixth equation used the equation about hypergeometric functions below:
$$
\int_0^x\,_2F_1(a,b;c,y)dy =x\,_3F_2(a,b,1;c,2,x).
$$
\end{proof}
Now by \eqref{loop44} and Lemma \ref{lemma1} and Lemma \ref{lemma2} we get \eqref{han}.

\subsection{Proof of  Cardy-Gamsa's formula }
In this section we  will prove the  equivalence of formula \eqref{cardy} and \eqref{han}. First we will recall some identities for the hypergeometric functions which will be used in our proof. We will assume that our hypergeometric functions
are all well defined. And they satisfies the following identities (see Chapter 8 of \cite{beals2010special}):
\begin{equation}\label{loop4}
\,_2F_1(a,b;c;x)=(1-x)^{-b} \HY(c-a,b;c;\frac{x}{x-1}).
\end{equation}
\begin{multline}\label{loop5}
\HY(a,b;c;x)=\frac{\Gamma(c)\Gamma(c-a-b)}{\Gamma(c-a)\Gamma(c-b)}\HY(a,b;a+b+1-c;1-x)\\
+\frac{\Gamma(c)\Gamma(a+b-c)}{\Gamma(a)\Gamma(b)}(1-x)^{c-a-b}\HY(c-a,c-b;c+1-a-b;1-x).
\end{multline}
\begin{equation}\label{loop6}
\,_2F_1(a,b;c;x)=(1-x)^{c-a-b} \HY(c-a,c-b;c;x).
\end{equation}
Notice that $\eta=\frac{\sigma}{\sigma-1}$ and $\sigma\in (0,1)$. We define a function $\phi$ on $[0,1]$ as follows:
\begin{multline}\label{function1}
\phi(t)=\frac{2\pi}{\sqrt{3}}+\frac{t}{t-1}\hy(1,\frac{4}{3},1;\frac{5}{3},2;\frac{t}{t-1})-(1-t)\hy(1,\frac{4}{3},1;\frac{5}{3},2;1-t)\\[3mm]
-2\log (1-t)-2\frac{\Gamma(\frac{2}{3})^2}{\Gamma(\frac{4}{3})}\sqrt[3]{\frac{t}{(t-1)^2}}\HY(1,\frac{2}{3};\frac{4}{3};\frac{t}{t-1}).
\end{multline}

To prove that \eqref{cardy} and \eqref{han} are equivalent,  we only need to show that $\phi(t)\equiv 0$. Notice
that(see Appendix) $\phi(0)=\frac{2\pi}{\sqrt{3}}-\hy(1,\frac{4}{3},1;\frac{5}{3},2;1)=0$, it is left  to show that $\phi'(t)\equiv 0$.
Let us take the following notations.
\begin{align*}
 &I(t):=\frac{t}{t-1}\hy(1,\frac{4}{3},1;\frac{5}{3},2;\frac{t}{t-1})-2\log (1-t),\\[3mm]
 &J(t):=-(1-t)\hy(1,\frac{4}{3},1;\frac{5}{3},2;1-t),\\[3mm]
 &K(t):=-2\frac{\Gamma(\frac{2}{3})^2}{\Gamma(\frac{4}{3})}\sqrt[3]{\frac{t}{(t-1)^2}}\HY(1,\frac{2}{3};\frac{4}{3};\frac{t}{t-1}).
\end{align*}
Define  $f(x)=x\hy(1,\frac{4}{3},1;\frac{5}{3},2;x)$.  It is easy to check that
$$
f'(x)=\HY(1,\frac{4}{3};\frac{5}{3};x).
$$
So
\begin{align}\label{loop7}
\frac{dI(t)}{dt}&=\frac{2}{1-t}+f'(\frac{t}{t-1})\frac{-1}{(1-t)^2}=\frac{2}{1-t}-\frac{1}{(1-t)^2}\HY(1,\frac{4}{3};\frac{5}{3};\frac{t}{t-1})\nonumber\\[3mm]
&=\frac{2}{1-t}-\frac{1}{1-t}\HY(1,\frac{1}{3};\frac{5}{3};t).
\end{align}
The last equation follows from \eqref{loop4} by assigning $a=\frac{1}{3}, b=1, c=\frac{5}{3}$. Similarly we  get
$$
\frac{dJ(t)}{dt}=f'(1-t)=\HY(1,\frac{4}{3};\frac{5}{3};1-t).
$$
Using \eqref{loop5} with $a=1, b=\frac{4}{3}, c=\frac{5}{3}$, we have
$$
\HY(1,\frac{4}{3};\frac{5}{3};1-t)=-\HY(1,\frac{4}{3};\frac{5}{3};t)+\frac{2}{3}\frac{\Gamma(\frac{2}{3})^2}{\Gamma(\frac{4}{3})}
t^{-\frac{2}{3}}\HY(\frac{1}{3},\frac{2}{3};\frac{1}{3};t)
$$
 By letting $a=\frac{1}{3}, b=\frac{2}{3}, c=\frac{1}{3}$ in \eqref{loop6}, the following holds
 $$
 \HY(\frac{1}{3},\frac{2}{3};\frac{1}{3};t)=(1-t)^{-\frac{2}{3}}\HY(0,-\frac{1}{3};\frac{1}{3},x)=(1-t)^{-\frac{2}{3}}.
 $$
 Therefore
 \begin{equation}\label{loop8}
 \frac{dJ(t)}{dt}=-\HY(1,\frac{4}{3};\frac{5}{3};t)+\frac{2}{3}\frac{\Gamma(\frac{2}{3})^2}{\Gamma(\frac{4}{3})}
(t(1-t))^{-\frac{2}{3}}.
 \end{equation}
Lastly we deal with the derivative of $K(t)$ with respect to $t$. By letting $a=\frac{1}{3}, b=\frac{2}{3}, c=\frac{4}{3}$ in \eqref{loop4}, we  get
$$
\HY(1,\frac{2}{3};\frac{4}{3};\frac{t}{t-1})=(1-t)^{\frac{2}{3}}\HY(1,\frac{2}{3};\frac{4}{3};t).
$$
Consequently,
$$
K(t)=-2\frac{\Gamma(\frac{2}{3})^2}{\Gamma(\frac{4}{3})}t^{\frac{1}{3}}\HY(1,\frac{2}{3};\frac{4}{3};t).
$$
And
\begin{align}\label{loop9}
\frac{dK(t)}{dt}&=-2\frac{\Gamma(\frac{2}{3})^2}{\Gamma(\frac{4}{3})}\big[\frac{1}{3}t^{-\frac{2}{3}}\HY(1,\frac{2}{3};\frac{4}{3};t)+t^{\frac{1}{3}}\frac{(1-t)^{-\frac{2}{3}}-\HY(1,\frac{2}{3};\frac{4}{3};t)}{3t}\big]\nonumber\\[3mm]
&=-\frac{2}{3}\frac{\Gamma(\frac{2}{3})^2}{\Gamma(\frac{4}{3})}t^{-\frac{2}{3}}(1-t)^{-\frac{2}{3}}.
\end{align}
Combining \eqref{loop7},\eqref{loop8} and \eqref{loop9}, we have
$$
\phi'(t)=\frac{dI(t)}{dt}+\frac{dJ(t)}{dt}+\frac{dK(t)}{t}=\frac{2}{1-t}-\frac{1}{1-t}\HY(1,\frac{1}{3};\frac{5}{3};t)-\HY(1,\frac{4}{3};\frac{5}{3};t).
$$
\begin{lemma}\label{lemma3}
$$
2-\HY(1,\frac{1}{3};\frac{5}{3};t)-(1-t)\HY(1,\frac{4}{3};\frac{5}{3};t)=0.
$$
\end{lemma}
\begin{proof}
By definition we have
$$
\HY(1,\frac{4}{3};\frac{5}{3};t)=1+\sum\limits_{n=1}^{\infty}\frac{\Gamma(n+\frac{4}{3})\Gamma(\frac{5}{3})}{\Gamma(\frac{4}{3})\Gamma(n+\frac{5}{3})}t^n.
$$
Therefore
$$
t\HY(1,\frac{4}{3};\frac{5}{3};t)=\sum\limits_{n=1}^{\infty}\frac{\Gamma(n+\frac{1}{3})\Gamma(\frac{5}{3})}{\Gamma(\frac{4}{3})\Gamma(n+\frac{2}{3})}t^n.
$$
Similarly
$$
\HY(1,\frac{1}{3};\frac{5}{3};t)=1+\sum\limits_{n=1}^{\infty}\frac{\Gamma(n+\frac{1}{3})\Gamma(\frac{5}{3})}{\Gamma(\frac{1}{3})\Gamma(n+\frac{5}{3})}t^n.
$$
By using the relation $\Gamma(x+1)=x\Gamma(x)$, we can see that the coefficient of $t^n$ in the sum is
$$
\frac{\Gamma(n+\frac{1}{3})\Gamma(\frac{5}{3})}{\Gamma(\frac{4}{3})\Gamma(n+\frac{2}{3})}-\frac{\Gamma(n+\frac{4}{3})\Gamma(\frac{5}{3})}{\Gamma(\frac{4}{3})\Gamma(n+\frac{5}{3})}-\frac{\Gamma(n+\frac{1}{3})\Gamma(\frac{5}{3})}{\Gamma(\frac{1}{3})\Gamma(n+\frac{5}{3})}=0.
$$
\end{proof}
From Lemma \ref{lemma3}, we have $\phi'(t)\equiv 0$, and therefore $\phi\equiv \phi(0)=0$. This completes the proof
of the equivalence between \eqref{cardy} and \eqref{han}.
\section{The other cases}
Given $z,w\in \H$, and $\gamma$ the
sample of the Brownian loop in the upper half plane. According to the property of Brownian path, almost surely, $z,w\not\in\gamma$. So except the case  that $\gamma$ disconnects both $z$ and $w$ from the boundary,  there are three other cases:

(1) $\gamma$ disconnects $z$ from the boundary but does not disconnect $w$ from the boundary;

(2) $\gamma$ disconnects $w$ from the boundary but does not disconnect $z$ from the boundary;

(3) $\gamma$  does neither disconnects $z$ from the boundary nor  disconnects $w$ from the boundary.

We will show that the total measure of above three cases are infinite. In fact, using the same method
as \cite{beliaev2013some},  we can show the following lemma.
\begin{lemma}
Suppose that $\gamma$ is the sample of the \SLEkk\frac{8}{3}/ from $0$ to $\epsilon$ and denote above three cases by
$E_1(z,w), E_2(z,w)$ and $E_3(z,w)$ respectively. Then
\begin{equation}
\prob[E_1(z,w)]=\frac{1}{4}\epsilon^2\big((\Im\frac{1}{z})^2-\Im\frac{1}{z}\Im\frac{1}{w} G(\sigma)\big)+O(\epsilon^3),
\end{equation}
\begin{equation}
\prob[E_2(z,w)]=\frac{1}{4}\epsilon^2\big[(\Im\frac{1}{w})^2-\Im\frac{1}{z}\Im\frac{1}{w} G(\sigma)\big]+O(\epsilon^3),
\end{equation}
\begin{equation}
\prob[E_3(z,w)]=1-\frac{1}{4}\epsilon^2\big[(\Im\frac{1}{w})^2+(\Im\frac{1}{z})^2-\Im\frac{1}{z}\Im\frac{1}{w} G(\sigma)\big]+O(\epsilon^3).
\end{equation}
\end{lemma}
The proof of this lemma is the same as in \cite{beliaev2013some}. We only need to prove that for \SLEkk\frac{8}{3}/  $\gamma$ from $0$
 to $\infty$, the following holds.
 $$
 \begin{array}{rl}
 \prob[\gamma&\text{passes the left of }z\text{ and the right of }w]\\[3mm]
 =&\frac{1}{4}(1-\frac{x}{|z|})(1+\frac{u}{|w|})(1-\frac{y}{|z|-x}\frac{v}{|w|+u}G(\sigma)).
 \end{array}
 $$
  $$
 \begin{array}{rl}
 \prob[\gamma&\text{passes the left of }w\text{ and the right of }z]\\[3mm]
 =&\frac{1}{4}(1+\frac{x}{|z|})(1-\frac{u}{|w|})(1-\frac{y}{|z|+x}\frac{v}{|w|-u}G(\sigma)).
 \end{array}
 $$
  $$
 \begin{array}{rl}
 \prob[\gamma&\text{passes the right of both }z\text{ and } w]\\[3mm]
 =&\frac{1}{4}(1-\frac{x}{|z|})(1-\frac{u}{|w|})(1+\frac{y}{|z|-x}\frac{v}{|w|-u}G(\sigma)).
 \end{array}
 $$
 where $G(\sigma)$ is the same as \eqref{loop11}.
 Then using the conformal map $F_\epsilon(z)=\frac{\epsilon z}{1+z}$ to convert the  \SLEkk\frac{8}{3}/ from $0$ to $\infty$ into the  \SLEkk\frac{8}{3}/ from $0$ to $\epsilon$. Combining the  above lemma and  the definition of the Brownian bubble measure and lemma \ref{lemma5},
 we can get
 $$
 \mu_{\H}^{\text{bub}}(0)(E_1(z,w))=\frac{1}{10}[(\frac{y}{x^2+y^2})^2-\frac{y}{x^2+y^2}\frac{v}{u^2+v^2}G(\sigma(z,w))].
 $$
 $$
 \mu_{\H}^{\text{bub}}(0)(E_2(z,w))=\frac{1}{10}[(\frac{v}{u^2+v^2})^2-\frac{y}{x^2+y^2}\frac{v}{u^2+v^2}G(\sigma(z,w))].
 $$
 $$
 \mu_{\H}^{\text{bub}}(0)(E_3(z,w))=\infty.
 $$
By relation \eqref{loop10} and calculating the integral on the upper half plane, we  see that the total mass of the Brownian loop measure of  these three sets are infinite. In fact, we can see intuitively that these three cases all contain the  loops with  arbitrary small diameter, while the event $E(z,w)$ in the main theorem exclude these small loops.
\section{Appendix: on the value of $ \hy(1,\frac{4}{3},1;2,\frac{5}{3};1)$ }
We want to prove that
\begin{align}\label{quzhi}
    \hy(1,\frac{4}{3},1;2,\frac{5}{3};1)=\frac{2\pi}{\sqrt{3}}.
\end{align}
\begin{lemma}\label{lemma7}
For $b>a>0$ we have
\begin{align}\label{quzhi2}
    \hy(1,1,a;2,b;1)=\frac{b-1}{a-1}(\psi(b-1)-\psi(b-a)).
\end{align}
where
\begin{align}
    \psi(x)=:\frac{\Gamma'(x)}{\Gamma(x)}
\end{align}
is the digamma function and $\Gamma(x)$ is the $\Gamma$-function.
\end{lemma}
Using Lemma \ref{lemma7}, we can get \eqref{quzhi}. In fact, let
$a=\frac{4}{3},b=\frac{5}{3}$ in \eqref{quzhi2}, we have
\begin{align}
   &\hy(1,\frac{4}{3},1;2,\frac{5}{3};1)=2\Big(\psi(\frac{2}{3})-\psi(\frac{1}{3})\Big)\nonumber=\frac{2\pi}{\sqrt{3}}.\\[3mm]
\end{align}\label{quzhi3}
This is because since for $\Gamma$ function we have
\begin{align}
\Gamma(x)\Gamma(1-x)=\frac{\pi}{\sin(\pi x)}
\end{align}
By differentiating two sides of \eqref{quzhi3}, we have
\begin{align}
\Gamma'(x)\Gamma(1-x)-\Gamma(x)\Gamma'(1-x)=-\frac{\pi^2\cos( \pi x)}{(\sin \pi x)^2}.
\end{align}
Let $x=\frac{2}{3}$, we  get
\begin{align}
\frac{\Gamma'(\frac{2}{3})}{\Gamma(\frac{2}{3})}-\frac{\Gamma'(\frac{1}{3})}{\Gamma(\frac{1}{3})}=\frac{1}{\Gamma(\frac{2}{3})\Gamma(\frac{1}{3})}\frac{\pi^2\frac{1}{2}}{3/4}=\frac{\pi}{\sqrt{3}}.
\end{align}
\textbf{Proof of Lemma \ref{lemma7}}
\begin{proof}
 We need the Thomae's result[see \cite{thomae1879ueber}]:
 \begin{lemma}[Thomae]\label{Thomae}
 If $s:=e+f-a-b-c$ and $\Re(a)>0, \Re(s)>0$, then
 \begin{align}\label{thomae2}
 \hy(a,b,c;e,f;1)=\frac{\Gamma(e)\Gamma(f)\Gamma(s)}{\Gamma(a)\Gamma(s+b)\Gamma(s+c)}\hy(e-a,f-a,s;s+b,s+c;1).
 \end{align}
 \end{lemma}
 Let $a=b=1, e=2$ in \eqref{thomae2}, we have $s=f-c$, so we have
 \begin{align}
   \hy(1,1,c;2,f;1)=&\frac{\Gamma(2)\Gamma(f)\Gamma(f-c)}{\Gamma(1)\Gamma(f-c+1)\Gamma(f)}\hy(1,f-1,f-c;f-c+1,f;1)\nonumber\\[3mm]
   &=\frac{1}{f-c}\sum\limits_{n=0}^\infty\frac{n!(f-1)_{n}(f-c)_{n}}{(f)_{n}(f-c+1)_{n}} \frac{1}{n!}\nonumber\\[3mm]
   &=\frac{1}{f-c}\sum\limits_{n=0}^\infty\frac{f-1}{f-1+n}\frac{f-c}{f-c+n}\nonumber\\[3mm]
   &=(f-1)\sum\limits_{n=0}^\infty\frac{1}{(f-1+n)(f-c+n)}.
 \end{align}
 In the third equation we used
 $$
 \frac{(x+1)_n}{(x)_n}=\frac{x+n}{x}.
 $$

 Also by the relation of the digamma function and Gamma function, we have
 \begin{align}
 \psi(z)=-\gamma+\sum\limits_{n=0}^\infty (\frac{1}{n}-\frac{1}{n+z}),\,\,z\neq0,-1,-2,-3,....
 \end{align}
 where $\gamma$ is  the Euler-Mascheroni constant. So we have
 \begin{align}
 &\psi(f-1)-\psi(f-c)=\sum\limits_{n=0}^\infty (\frac{1}{n+f-c}-\frac{1}{n+f-1})\nonumber\\[3mm]
 =&(c-1)\sum\limits_{n=0}^\infty \frac{1}{(n+f-c)(n+f-1)}\nonumber\\[3mm]
 =&\frac{c-1}{f-1}\hy(1,1,c;2,f;1).
 \end{align}
 This finishes the proof of Lemma \ref{lemma1}.
 \end{proof}
\bibliography{loop}

\begin{thebibliography}{10}

\bibitem{beliaev2013some}
Dmitry Beliaev and Fredrik~Johansson Viklund.
\newblock Some remarks on {SLE} bubbles and schramm's two-point observable.
\newblock {\em Communications in Mathematical Physics}, 320(2):379--394, 2013.

\bibitem{lawler2003conformal}
Gregory Lawler, Oded Schramm, and Wendelin Werner.
\newblock Conformal restriction: the chordal case.
\newblock {\em Journal of the American Mathematical Society}, 16(4):917--955,
  2003.

\bibitem{gamsa2006correlation}
Adam Gamsa and John Cardy.
\newblock Correlation functions of twist operators applied to single
  self-avoiding loops.
\newblock {\em Journal of Physics A: Mathematical and General}, 39(41):12983,
  2006.

\bibitem{lawler2004brownian}
Gregory~F Lawler and Wendelin Werner.
\newblock The brownian loop soup.
\newblock {\em Probability theory and related fields}, 128(4):565--588, 2004.

\bibitem{rohde2011basic}
Steffen Rohde and Oded Schramm.
\newblock Basic properties of {SLE}.
\newblock In {\em Selected Works of Oded Schramm}, pages 989--1030. Springer,
  2011.

\bibitem{lawler2011conformal}
Gregory~F Lawler, Oded Schramm, and Wendelin Werner.
\newblock Conformal invariance of planar loop-erased random walks and uniform
  spanning trees.
\newblock In {\em Selected Works of Oded Schramm}, pages 931--987. Springer,
  2011.

\bibitem{schramm2001percolation}
Oded Schramm et~al.
\newblock A percolation formula.
\newblock {\em Electron. Comm. Probab}, 6:115--120, 2001.

\bibitem{lawler2008conformally}
Gregory~F Lawler.
\newblock {\em Conformally invariant processes in the plane}.
\newblock Number 114. American Mathematical Soc., 2008.

\bibitem{beals2010special}
Richard Beals and Roderick Wong.
\newblock {\em Special functions: a graduate text}, volume 126.
\newblock Cambridge University Press, 2010.

\bibitem{thomae1879ueber}
J~Thomae.
\newblock Ueber die functionen, welche durch reihen von der form dargestellt
  werden....
\newblock {\em Journal f{\"u}r die Reine und angewandte Mathematik}, 87:26--73,
  1879.

\end{thebibliography}
\bibliographystyle{unsrt}
\end{document}